\documentclass[runningheads,a4paper,11pt]{llncs}
\usepackage[margin=2cm]{geometry}

\usepackage[caption=false,font=footnotesize]{subfig}
\usepackage{amsmath,amssymb,amsfonts}
\usepackage{graphicx,url,cite,algorithm,booktabs,algpseudocode}
\usepackage{hyperref}
\usepackage[usenames,dvipsnames]{xcolor}

\algnotext{EndFor}
\algnotext{EndIf}
\algnotext{EndWhile}
 
\newcommand{\vc}[1]{\boldsymbol{#1}}

\newcommand{\eps}{\varepsilon}
\newcommand{\No}{\mathcal{N}}
\newcommand{\HNo}{\mathcal{H}}
\renewcommand{\Pr}{\mathbb{P}}
\newcommand{\expn}{\mathbb{E}}
\newcommand{\cD}{\mathcal{D}}

\begin{document}

\mainmatter

\title{Hypercube LSH for approximate near neighbors}
\titlerunning{Hypercube LSH for approximate near neighbors}
\toctitle{Lecture Notes in Computer Science}

\author{Thijs Laarhoven}
\authorrunning{Thijs Laarhoven}
\tocauthor{Thijs Laarhoven}

\institute{IBM Research\\R\"{u}schlikon, Switzerland\\ \texttt{mail@thijs.com}}

\maketitle

\begin{abstract}
A celebrated technique for finding near neighbors for the angular distance involves using a set of \textit{random} hyperplanes to partition the space into hash regions [Charikar, STOC 2002]. Experiments later showed that using a set of \textit{orthogonal} hyperplanes, thereby partitioning the space into the Voronoi regions induced by a hypercube, leads to even better results [Terasawa and Tanaka, WADS 2007]. However, no theoretical explanation for this improvement was ever given, and it remained unclear how the resulting hypercube hash method scales in high dimensions.

In this work, we provide explicit asymptotics for the collision probabilities when using hypercubes to partition the space. Vectors at angle $\frac{\pi}{3}$ collide with probability $(\frac{\sqrt{3}}{\pi})^{d + o(d)}$ in dimension $d$, compared to $(\frac{2}{3})^d$ for random hyperplanes, and near-parallel vectors collide with similar asymptotic probabilities in both cases. The limiting exponential base tends continuously to zero as the angle approaches $\frac{\pi}{2}$ from below.

For fixed near and far angular thresholds below $\frac{\pi}{2}$, these asymptotics determine the limiting LSH collision-probability ratio $\rho$. We further describe how to use hypercube LSH in practice. A direct partial-hypercube analysis gives a heuristic lattice-sieving time bound $2^{0.3347d+o(d)}$.

\keywords{(approximate) near neighbors, locality-sensitive hashing, large deviations, dimensionality reduction, lattice algorithms}
\end{abstract}


\paragraph*{Correction to the earlier version.} This revision corrects the third branch of Theorem~\ref{thm:main} and its proof, as well as the consequences that used that branch. It also replaces an unjustified transfer to partial hypercubes by a direct analysis and corrects the resulting sieve estimate. The incorrect branch appeared in the earlier arXiv version and the MFCS 2017 publication. It is retained as a dashed comparison curve in Figure~\ref{fig:main}.

\section{Introduction}

\paragraph*{Finding (approximate) near neighbors.} A key computational problem in various research areas, including machine learning, pattern recognition, data compression, coding theory, and cryptanalysis~\cite{shakhnarovich05, bishop06, dubiner10, duda00, may15, laarhoven15crypto}, is finding near neighbors: given a data set $\cD \subset \mathbb{R}^d$ of cardinality $n$, design a data structure and preprocess $\cD$ in a way that, when given a query vector $\vc{q} \in \mathbb{R}^d$, one can efficiently find a near point to $\vc{q}$ in $\cD$. Due to the ``curse of dimensionality''~\cite{indyk98} this problem is known to be hard to solve exactly (in the worst case) in high dimensions $d$, so a common relaxation of this problem is the \textit{$(c,r)$-approximate near neighbor problem} ($(c,r)$-ANN): given that the nearest neighbor lies at distance at most $r$ from $\vc{q}$, design an algorithm that finds an element $\vc{p} \in \cD$ at distance at most $c \cdot r$ from $\vc{q}$. 

\paragraph{Locality-sensitive hashing (LSH) and filtering (LSF).} A prominent class of algorithms for finding near neighbors in high dimensions is formed by locality-sensitive hashing (LSH)~\cite{indyk98} and locality-sensitive filtering (LSF)~\cite{becker16lsf}. These solutions are based on using regions of space, in a way that nearby vectors have a higher probability of ending up in the same hash region than distant vectors. By carefully tuning (i) the number of hash regions per hash table, and (ii) the number of randomized hash tables, one can then guarantee that with high probability (a) nearby vectors will \textit{collide} in at least one of the hash tables, and (b) the expected number of distant collisions remains small. For LSH, a simple lookup in all of $\vc{q}$'s hash buckets then provides a fast way of finding near neighbors to $\vc{q}$, while for LSF the lookups are slightly more involved. For various metrics, LSH and LSF currently provide the best performance in high dimensions~\cite{andoni15, becker16lsf, andoni17, christiani17}.

\paragraph*{Near neighbors on the sphere.} 
In this work we will focus on the near neighbor problem under the \textit{angular distance}, where two vectors $\vc{x}, \vc{y}$ are considered nearby iff their common angle $\theta$ is small~\cite{charikar02, sundaram13, schmidt14, andoni15cp}. This equivalently corresponds to near neighbor searching for the $\ell_2$-norm, where the entire data set is assumed to lie on a sphere. A special case of $(c,r)$-ANN on the sphere, often considered in the literature, is the \textit{random} case $r = \frac{1}{c}\sqrt{2}$ and $c \cdot r = \sqrt{2}$, in part due to a reduction from near neighbor under the Euclidean metric for general data sets to $(c,r)$-ANN on the sphere with these parameters~\cite{andoni15}. 

\subsection{Related work}

\paragraph*{Upper bounds.} Perhaps the most well-known and widely used solution for ANN for the angular distance is Charikar's hyperplane LSH~\cite{charikar02}, where a set of \textit{random} hyperplanes is used to partition the space into regions. Due to its low computational complexity and the simple form of the collision probabilities (with no hidden order terms in $d$), this method is easy to instantiate in practice and commonly achieves the best performance out of all LSH methods when $d$ is not too large. For large $d$, both spherical cap LSH~\cite{andoni14, andoni15} and cross-polytope LSH~\cite{terasawa07, eshghi08, andoni15cp, kennedy17} are known to perform better than hyperplane LSH. Experiments from~\cite{terasawa07, terasawa09} showed that using \textit{orthogonal} hyperplanes, partitioning the space into Voronoi regions induced by the vertices of a hypercube, also leads to superior results compared to hyperplane LSH; however, no theoretical guarantees for the resulting hypercube LSH method were given, and it remained unclear whether the improvement persists in high dimensions.

\paragraph{Lower bounds.} For the case of \textit{random} data sets, lower bounds have also been found, matching the performance of spherical cap and cross-polytope LSH for large $c$~\cite{motwani07, odonnell11, andoni15cp}. These lower bounds are commonly in a model where it is assumed that collision probabilities are ``not too small'', and in particular not exponentially small in $d$. Therefore it is not clear whether one can further improve upon cross-polytope LSH when the number of hash regions is exponentially large, which would for instance be the case for hypercube LSH. Together with the experimental results from~\cite{terasawa07, terasawa09}, this naturally begs the question: how efficient is hypercube LSH? Is it better than hyperplane LSH and/or cross-polytope LSH? And how does hypercube LSH compare to other methods in practice?

\subsection{Contributions}

\paragraph{Hypercube LSH.} By carefully analyzing the collision probabilities for hypercube LSH using results from large deviations theory, we show that hypercube LSH is indeed different from, and superior to hyperplane LSH for large $d$. The following main theorem states the asymptotic form of the collision probabilities when using hypercube LSH, which are also visualized in Figure~\ref{fig:main} in comparison with hyperplane LSH.

\begin{theorem} [Collision probabilities for hypercube LSH] \label{thm:main}
Let $\vc{X}, \vc{Y} \sim \No(0,1)^d$, let $\theta \in [0, \pi]$ denote the angle between $\vc{X}$ and $\vc{Y}$, and let $p(\theta)$ denote the probability that $\vc{X}$ and $\vc{Y}$ are mapped to the same hypercube hash region. For $\theta \in (0, \arccos \frac{2}{\pi})$ (respectively $\theta \in (\arccos \frac{2}{\pi}, \frac{\pi}{3})$), let $\beta_0 \in (1, \infty)$ (resp.\ $\beta_1 \in (1, \infty)$) be the unique solution to:
\begin{align}
\arccos\left(\frac{-1}{\beta_0}\right) = \frac{(\beta_0 - \cos \theta) \sqrt{\beta_0^2 - 1}}{\beta_0 (\beta_0 \cos \theta - 1)} \, , \qquad \arccos\left(\frac{1}{\beta_1}\right) = \frac{(\beta_1 + \cos \theta) \sqrt{\beta_1^2-1}}{\beta_1 (\beta_1 \cos \theta + 1)} \, . \label{eq:beta}
\end{align}
%
For $\theta \in (\frac{\pi}{3},\frac{\pi}{2})$, let $\beta_2 \in (0,1)$ be the unique solution to
\begin{align}
\operatorname{arcosh}\left(\frac{1}{\beta_2}\right)
= \frac{(\beta_2+\cos\theta)\sqrt{1-\beta_2^2}}
{\beta_2(1+\beta_2\cos\theta)}. \label{eq:beta2}
\end{align}
Then, for each fixed $\theta$ as $d$ tends to infinity, $p(\theta)$ satisfies:
\begin{align}
p(\theta) = \begin{cases}
\left(\displaystyle\frac{(\beta_0 - \cos \theta)^2 }{\pi \beta_0 (\beta_0 \cos \theta - 1) \sin \theta}\right)^{d + o(d)}, & \qquad \text{if } \theta \in [0, \arccos \tfrac{2}{\pi}]; \\[3ex]
\left(\displaystyle\frac{(\beta_1 + \cos \theta)^2}{\pi \beta_1 (\beta_1 \cos \theta + 1) \sin \theta}\right)^{d + o(d)}, & \qquad \text{if } \theta \in [\arccos \tfrac{2}{\pi}, \tfrac{\pi}{3}]; \\[3ex]
\left(\displaystyle\frac{(\beta_2 + \cos \theta)^2}{\pi \beta_2 (\beta_2 \cos \theta + 1) \sin \theta}\right)^{d + o(d)}, & \qquad \text{if } \theta \in [\tfrac{\pi}{3}, \tfrac{\pi}{2}); \\[3ex]
0, & \qquad \text{if } \theta \in [\tfrac{\pi}{2}, \pi].
\end{cases}
\end{align}
The expressions at $\theta=0$, $\arccos(2/\pi)$, and $\pi/3$ are understood by continuity; in particular, $p(0)=1$ and $\beta_1=\beta_2=1$ at $\pi/3$.
\end{theorem}

\begin{figure}[!t]
\begin{center}\includegraphics[width=14cm]{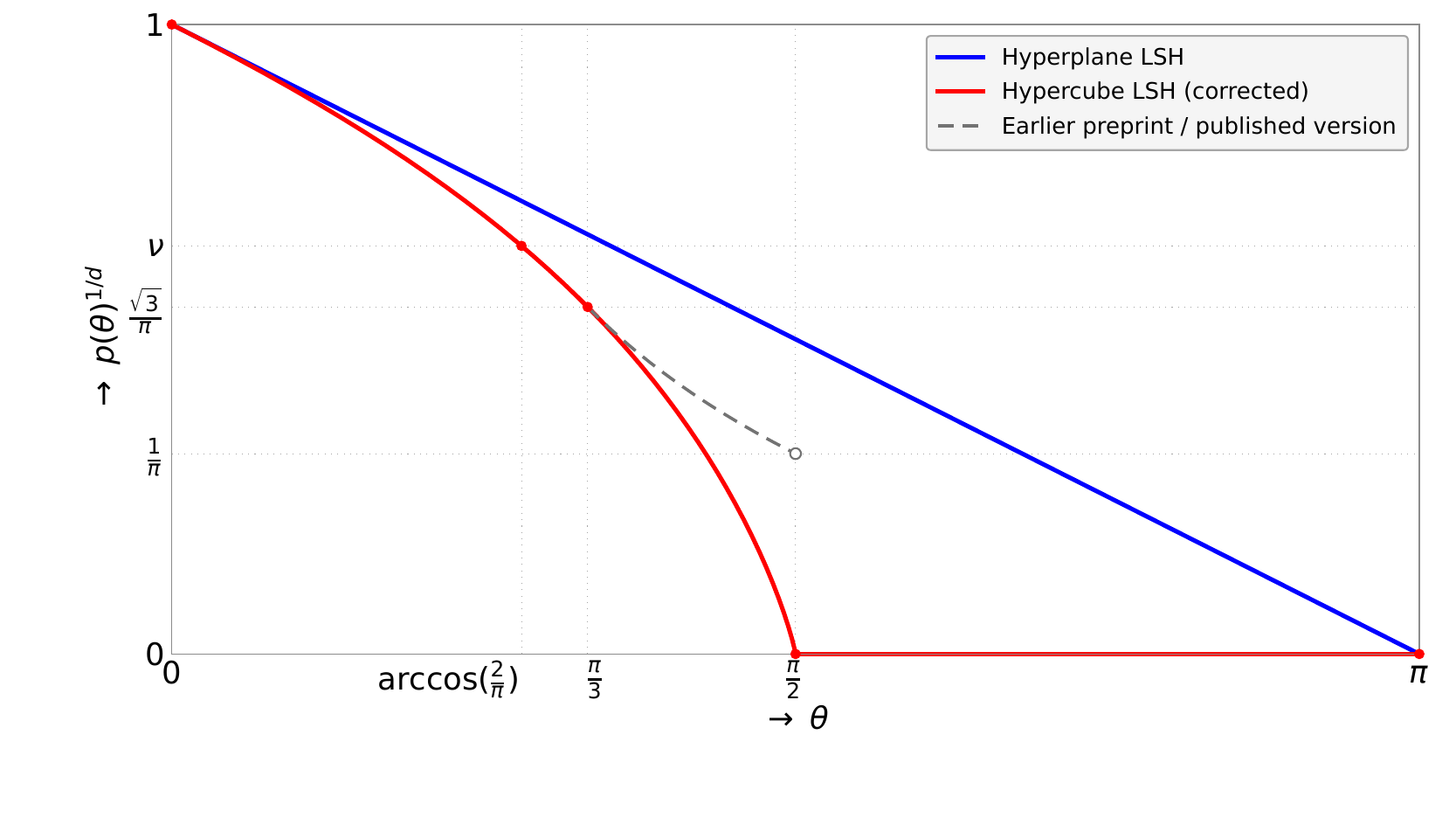}\end{center}
\caption{Asymptotics of collision probabilities for hypercube LSH, compared to hyperplane LSH. Here $\nu = \pi/(2 \sqrt{\pi^2 - 4})$, and the dashed vertical lines correspond to boundary points of the piecewise parts of Theorem~\ref{thm:main}. The blue line indicates hyperplane LSH with $d$ random hyperplanes. The red curve is the corrected hypercube limit; the gray dashed curve is the incorrect third branch from the earlier arXiv and published versions, shown only for comparison. \label{fig:main}}
\end{figure}

For fixed near and far angles below $\pi/2$, Theorem~\ref{thm:main} determines limiting ratios $\rho$ of logarithmic collision probabilities; see Section~3.2. Figure~\ref{fig:rho} compares these ratios for hyperplane, hypercube and cross-polytope LSH at a fixed far angle $\theta_2 = 5\pi/12$. These are asymptotics for $d \to \infty$, and may not accurately reflect performance for moderate $d$. We further briefly discuss how the hashing for hypercube LSH can be made efficient.

\begin{figure}[!t]
\begin{center}
\includegraphics[width=\textwidth]{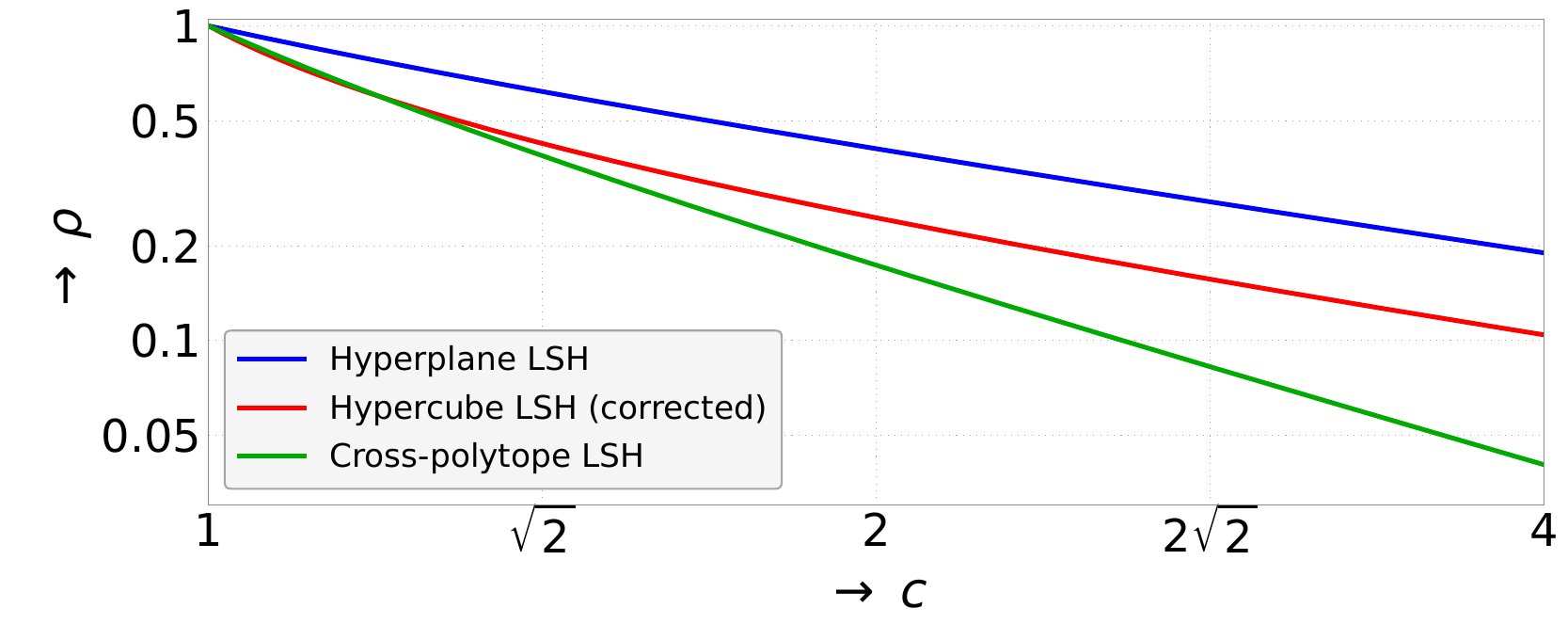}
\end{center}
\caption{Asymptotics for the LSH collision-probability ratio $\rho$ for hyperplane, hypercube, and cross-polytope LSH at fixed $\theta_2=5\pi/12$, with $c\cdot r=2\sin(5\pi/24)$ and $\theta_1=2\arcsin(\sin(5\pi/24)/c)$. The hyperplane ratio is exact for arbitrary $d$; the other curves omit terms vanishing as $d\to\infty$.\label{fig:rho}}
\end{figure}

\paragraph*{Partial hypercube LSH.} As the number of hash regions of a full-dimensional hypercube is often prohibitively large, we also consider \textit{partial} hypercube LSH, where a $d'$-dimensional hypercube is used to partition a data set in dimension $d$. We characterize the small-$d'$ regime where the exponential collision rates agree with hyperplane LSH, and give a direct rate formula for each fixed ratio $d'/d\in(0,1)$. We also retain the empirical comparison for larger $d'$.

\paragraph*{Application to lattice sieving.} Finally, we consider a specific use case of different LSH methods, in the context of lattice cryptanalysis. Using the actual partial-hypercube rates, we obtain a heuristic time bound $2^{0.3347d+o(d)}$, slightly improving upon hyperplane hashing. We discuss the previous experiments for this application.


\section{Preliminaries}
\label{sec:pre}

\paragraph*{Notation.} We denote probabilities with $\Pr(\cdot)$ and expectations with $\expn(\cdot)$. Capital letters commonly denote random variables, and boldface letters denote vectors. We informally write $\Pr(X = x)$ for continuous $X$ to denote the density of $X$ at $x$. For probability distributions $\mathcal{D}$, we write $X \sim \mathcal{D}$ to denote that $X$ is distributed according to $\mathcal{D}$. For sets $S$, with abuse of notation we further write $X \sim S$ to denote $X$ is drawn uniformly at random from $S$. We write $\No(\mu, \sigma^2)$ for the normal distribution with mean $\mu$ and variance $\sigma^2$, and $\HNo(\mu, \sigma^2)$ for the distribution of $|X|$ when $X \sim \No(\mu, \sigma^2)$. For $\mu = 0$ the latter corresponds to the \textit{half-normal} distribution. We write $\vc{X} \sim \mathcal{D}^d$ to denote a $d$-dimensional vector where each entry is independently distributed according to $\mathcal{D}$. In what follows, $\|\vc{x}\| = \sqrt{\sum_i x_i^2}$ denotes the Euclidean norm, and $\langle \vc{x}, \vc{y} \rangle = \sum_i x_i y_i$ denotes the standard inner product. We denote the angle between two vectors by $\phi(\vc{x}, \vc{y}) = \arccos \langle \vc{x} / \|\vc{x}\|, \vc{y} / \|\vc{y}\| \rangle$. 

\begin{lemma}[Distribution of angles between random vectors {\cite[Lemma 2]{becker16lsf}}] \label{lem:sphericalcap}
Let $\vc{X},\vc{Y}\sim\No(0,1)^d$ be independent. Their angle has density
\begin{align}
g_d(\theta)=\frac{\Gamma(d/2)}{\sqrt\pi\,\Gamma((d-1)/2)}(\sin\theta)^{d-2},\qquad 0<\theta<\pi.
\end{align}
For each fixed $\theta\in(0,\pi)$, $d^{-1}\ln g_d(\theta)\to\ln\sin\theta$.
\end{lemma}


\paragraph*{Locality-sensitive hashing.} Locality-sensitive hash functions~\cite{indyk98} are functions $h$ mapping a $d$-dimensional vector $\vc{x}$ to a low-dimensional \textit{sketch} $h(\vc{x})$, such that vectors which are nearby in $\mathbb{R}^d$ are more likely to be mapped to the same sketch than distant vectors. For the angular distance\footnote{This is a metric on the unit sphere and a pseudometric on nonzero vectors.} $\phi(\vc{x}, \vc{y})$, we quantify a set of hash functions $\mathcal{H}$ as follows (see~\cite{indyk98}):

\begin{definition} 
A hash family $\mathcal{H}$ is called $(\theta_1, \theta_2, p_1, p_2)$-sensitive if for nonzero $\vc{x}, \vc{y} \in \mathbb{R}^d$ we have:
\begin{itemize}
  \item If $\phi(\vc{x}, \vc{y}) \leq \theta_1$ then $\Pr_{h \sim \mathcal{H}}(h(\vc{x}) = h(\vc{y})) \geq p_1$;
  \item If $\phi(\vc{x}, \vc{y}) \geq \theta_2$ then $\Pr_{h \sim \mathcal{H}}(h(\vc{x}) = h(\vc{y})) \leq p_2$.
\end{itemize}
\end{definition}

The existence of locality-sensitive hash families implies the existence of fast algorithms for (approximate) near neighbors, as the following lemma describes\footnote{Various conditions and order terms (which are commonly $n^{o(1)}$) are omitted here for brevity.}. For more details on the general principles of LSH, we refer the reader to e.g.~\cite{indyk98, andoni09}.

\begin{lemma}[Locality-sensitive hashing {\cite{indyk98}}]  
\label{lem:lsh}
Suppose there exists a $(\theta_1,\theta_2,p_1,p_2)$-sensitive family $\mathcal H$ with $0<p_2<p_1<1$, and let $\rho=\log(p_1)/\log(p_2)$. Put
\begin{align}
k=\max\left\{1,\left\lceil\frac{\ln n}{\ln(1/p_2)}\right\rceil\right\},
\qquad L=\left\lceil p_1^{-k}\ln n\right\rceil.
\end{align}
Using $L$ independent tables, each concatenating $k$ hashes, we can w.h.p.\ either find an element of the data set at angle at most $\theta_2$ from $\vc q$, or conclude that none lies at angle at most $\theta_1$. Apart from hash-evaluation and distance-computation costs, the expected query time is $O(L(1+np_2^k))$, with $O(nL)$ space and preprocessing. In particular these bounds are at most $p_1^{-1}n^{\rho+o(1)}$ and $p_1^{-1}n^{1+\rho+o(1)}$, respectively. The factor $p_1^{-1}$ can be absorbed into $n^{o(1)}$ when $p_1=n^{-o(1)}$, but not in general when collision probabilities are exponentially small in $d$.
\end{lemma}

\paragraph*{Hyperplane LSH.} For the angular distance, Charikar~\cite{charikar02} introduced the hash family $\mathcal{H} = \{h_{\vc{a}}: \vc{a} \sim \mathcal{D}\}$ where $\mathcal{D}$ is any spherically symmetric distribution on $\mathbb{R}^d$, and $h_{\vc{a}}$ satisfies:
\begin{align}
h_{\vc{a}}(\vc{x}) = \begin{cases} +1, & \text{if } \langle \vc{a}, \vc{x}\rangle \geq 0; \\ -1, & \text{if } \langle \vc{a}, \vc{x}\rangle < 0. \end{cases}
\end{align}
The vector $\vc{a}$ can be interpreted as the normal vector of a random hyperplane, and the hash value depends on which side of the hyperplane $\vc{x}$ lies on. For this hash function, the probability of a collision depends linearly on the angle between $\vc{x}$ and $\vc{y}$:
\begin{align}
\Pr_{h \sim \mathcal{H}}\big(h(\vc{x}) = h(\vc{y})\big) 
= 1 - \frac{\phi(\vc{x}, \vc{y})}{\pi}\, .
\end{align}
For any two angles $\theta_1 < \theta_2$, the above family $\mathcal{H}$ is $(\theta_1, \theta_2, 1 - \frac{\theta_1}{\pi}, 1 - \frac{\theta_2}{\pi})$-sensitive. 

\paragraph{Large deviations theory.} Let $\{\vc{Z}_d\}_{d \in \mathbb{N}} \subset \mathbb{R}^k$ be a sequence of random vectors corresponding to an empirical mean, i.e.\ $\vc{Z}_d = \frac{1}{d} \sum_{i=1}^d \vc{U}_i$ with $\vc{U}_i$ i.i.d. We define the logarithmic moment generating function $\Lambda$ of $\vc{Z}_d$ as:
\begin{align}
\Lambda(\vc{\lambda}) = \ln \expn_{\vc{U}_1} \left[\exp \langle \vc{\lambda}, \vc{U}_1 \rangle\right].
\end{align}
Define $\mathcal{D}_{\Lambda} = \{\vc{\lambda} \in \mathbb{R}^k: \Lambda(\vc{\lambda}) < \infty\}$. The \textit{Fenchel-Legendre} transform of $\Lambda$ is defined as:
\begin{align}
\Lambda^*(\vc{z}) = \sup_{\vc{\lambda} \in \mathbb{R}^k} \left\{\langle \vc{\lambda}, \vc{z} \rangle - \Lambda(\vc{\lambda})\right\}.
\end{align}
The following result describes that under certain conditions on $\{\vc{Z}_d'\}$, the asymptotics of the probability measure on a set $F$ are related to the function $\Lambda^*$. 

\begin{lemma} [Cram\'er's theorem {\cite[Chapter 2]{dembo10}}] \label{lem:ge}
Let $\vc{0}$ be contained in the interior of $\mathcal{D}_{\Lambda}$, and let $\vc{Z}_d$ be an empirical mean of i.i.d. random vectors. For every open set $G$ and closed set $F$,
\begin{align}
-\inf_{\vc{z}\in G}\Lambda^*(\vc{z})
&\leq \liminf_{d\to\infty}\frac{1}{d}\ln\Pr(\vc{Z}_d\in G),\\
\limsup_{d\to\infty}\frac{1}{d}\ln\Pr(\vc{Z}_d\in F)
&\leq -\inf_{\vc{z}\in F}\Lambda^*(\vc{z}).
\end{align}
\end{lemma}

When the infima over the interior and closure of a set agree, these bounds give its probability up to a factor $\exp(o(d))$. Point probabilities for continuous random variables cannot be substituted into this statement; the proof below uses small neighborhoods and angular intervals.


\section{Hypercube LSH}
\label{sec:full}

In this section, we will analyze full-dimensional hypercube hashing, with hash family $\mathcal{H} = \{h_A: A \in SO(d)\}$ where $SO(d) \subset \mathbb{R}^{d \times d}$ denotes the rotation group, and $h_A$ satisfies:
\begin{align}
h_A(\vc{x}) = (h_1(A \vc{x}), \dots, h_d(A \vc{x})), \qquad \quad h_i(\vc{x}) = \begin{cases} +1, & \text{if } x_i \geq 0; \\ -1, & \text{if } x_i < 0. \end{cases} 
\end{align}
In other words, a hypercube hash function first applies a uniformly random rotation, and then maps the resulting vector to the orthant it lies in. This equivalently corresponds to a concatenation of $d$ hyperplane hash functions, where all hyperplanes are orthogonal. Collision probabilities for prescribed angles $\theta$ between $\vc{x}$ and $\vc{y}$ are denoted by:
\begin{align}
p(\theta) = \Pr(h_A(\vc{x}) = h_A(\vc{y}) \ | \ \phi(\vc{x}, \vc{y}) = \theta).
\end{align}
Above, the randomness is over $h_A \sim \mathcal{H}$, with $\vc{x}$ and $\vc{y}$ arbitrary vectors at angle $\theta$ (e.g.\ $\vc{x} = \vc{e}_1$ and $\vc{y} = \vc{e}_1 \cos \theta + \vc{e}_2 \sin \theta$). Alternatively, the random rotation $A$ inside $h_A$ may be omitted, and the probability can be computed over $\vc{X}, \vc{Y}$ drawn uniformly at random from a spherically symmetric distribution, \textit{conditioned} on their common angle being $\theta$.

\subsection{Outline of the proof of Theorem~\ref{thm:main}}

Although Theorem~\ref{thm:main} is a key result, due to space restrictions we have decided to defer the full proof to the appendix. The approach of the proof can be summarized by the following four steps:
\begin{itemize}
\item Rewrite the collision probabilities in terms of (normalized) half-normal vectors $\vc{X}, \vc{Y}$;
\item Introduce dummy variables $x, y$ for the norms of these half-normal vectors, so that the probability can be rewritten in terms of unnormalized half-normal vectors;
\item Apply Cram\'er's theorem (Lemma~\ref{lem:ge}) to the three-dimensional vector given by \\ $\vc{Z} = \frac{1}{d}(\sum_i X_i Y_i, \sum_i X_i^2, \sum_i Y_i^2)$ to compute the resulting probabilities for arbitrary $x, y$;
\item Maximize the resulting expressions over $x, y > 0$ to get the final result.
\end{itemize}
The majority of the technical part of the proof lies in computing $\Lambda^*(\vc{z})$, which involves a somewhat tedious optimization of a multivariate function through a case-by-case analysis. 

\paragraph{A note on Gaussian approximations.} From the (above outline of the) proof, and the observation that the final optimization over $x, y$ yields $x = y = 1$ as the optimum, one might wonder whether a simpler analysis might be possible by assuming (half-)normal vectors are already normalized. Such a computation however would only lead to an approximate solution, which is perhaps easiest to see by computing collision probabilities for $\theta = 0$. In the exact computation, where vectors are normalized, $\langle \vc{X}, \vc{Y} \rangle = 1$ implies $\vc{X} = \vc{Y}$. If however we do not take into account the norms of $\vc{X}$ and $\vc{Y}$, and do not condition on the norms being equal to $1$, then $\langle \vc{X}, \vc{Y} \rangle = 1$ could also mean that $\vc{X}, \vc{Y}$ are slightly longer than $1$ and have a small, non-zero angle. In fact, such a computation would indeed yield $p(\theta)^{1/d} \not\to 1$ as $\theta \to 0$.

\subsection{Consequences of Theorem~\ref{thm:main}}

From Theorem~\ref{thm:main}, we can draw several conclusions. Substituting values for $\theta$, we find $p(\frac{\pi}{3})^{1/d} = \frac{\sqrt{3}}{\pi} + o(1)$, while $p(\frac{\pi}{2}) = 0$ in every dimension. The limiting exponential base $B(\theta)$ is continuous, including at $\theta = \frac{\pi}{2}$. More precisely, as $\theta\to(\frac{\pi}{2})^-$,
\begin{align}
B(\theta) \sim \frac{\cos\theta}{\pi}\ln\frac{1}{\cos\theta} \longrightarrow 0.
\end{align}
To understand the boundary $\theta = \arccos \frac{2}{\pi}$ of the piece-wise limit function, note that two (normalized) half-normal vectors $\vc{X}, \vc{Y}$ have expected inner product tending to $\frac{2}{\pi}$.

\paragraph{LSH exponents $\rho$ for fixed angular thresholds.} For fixed $0 < \theta_1 < \theta_2 < \frac{\pi}{2}$, Theorem~\ref{thm:main} gives
\begin{align}
\rho_d = \frac{\ln p(\theta_1)}{\ln p(\theta_2)} \longrightarrow \frac{\ln B(\theta_1)}{\ln B(\theta_2)}.
\end{align}
For example, with $\theta_2 = 5\pi/12$ and Euclidean approximation factor $c = 2$, the near angle is $\theta_1 = 2\arcsin(\sin(5\pi/24)/2)$, giving $\rho \approx 0.2454$ for hypercube LSH and $\rho \approx 0.4068$ for hyperplane LSH. Figure~\ref{fig:rho} compares these limiting ratios with those for cross-polytope LSH.

For fixed $\theta_1$, this limiting ratio tends to zero as $\theta_2\to(\frac{\pi}{2})^-$. This is an iterated limit of fixed-angle collision-probability ratios. Theorem~\ref{thm:main} does not give uniform estimates when the angle varies with $d$, nor by itself a query-time bound for random data near orthogonality. In particular, the positive limiting ratios previously stated for that setting do not follow from the corrected theorem. Applying the usual LSH query-time formula also requires accounting for the exponentially small collision probabilities and the number of hash tables.

\paragraph{Scaling at $\theta \to 0$ and fixed-threshold asymptotics for large $c$.} For $\theta$ close to $0$, by Theorem~\ref{thm:main} we are in the regime defined by $\beta_0$. For $\cos \theta = 1 - \eps$ with $\eps > 0$ small, observe that $\beta_0 \approx 1$ satisfies $\beta_0 > 1/\cos \theta$. Computing a Taylor expansion around $\eps = 0$, we eventually find $\beta_0 = 1 + \eps + \frac{2 \sqrt{2}}{\pi} \eps^{3/2} + O(\eps^2)$. Substituting this value $\beta_0$ into $p(\theta)$ with $\cos \theta = 1 - \eps$, we find:
\begin{align}
p(\theta) 
= \left(1 - \frac{\sqrt{2}}{\pi} \sqrt{\eps} + O(\eps)\right)^{d + o(d)}.
\end{align}
To compare this with hyperplane LSH, recall that the collision probability for $d$ random hyperplanes is equal to $(1 - \frac{\theta}{\pi})^d$. Since $\cos \theta = 1 - \eps$ translates to $\theta = \sqrt{2 \eps} (1 + O(\eps))$, the collision probabilities for hyperplane hashing in this regime are also $(1 - \frac{\sqrt{2}}{\pi} \sqrt{\eps} + O(\eps))^d$. In other words, for angles $\theta \to 0$, the collision probabilities for hyperplane hashing and hypercube hashing are similar. This can also be observed in Figure~\ref{fig:main}. For a fixed far angle $\theta_2 < \pi/2$, put $r_2 = 2\sin(\theta_2/2)$ and $\theta_1 = 2\arcsin(r_2/(2c))$. Taking $d\to\infty$ first and then $c\to\infty$ gives
\begin{align}
\rho \sim \frac{r_2}{\pi c\,[-\ln B(\theta_2)]}.
\end{align}
For hyperplane LSH, the denominator $-\ln B(\theta_2)$ is replaced by $-\ln(1-\theta_2/\pi)$. Thus the comparison depends on the fixed far threshold; there is no universal factor $\log_2\pi$ obtained by taking $\theta_2\to\pi/2$. For fixed $\theta_2$, cross-polytope LSH has $\rho = O(1/c^2)$ for large $c$~\cite{andoni15cp}.

For fixed $c>1$, a different comparison arises when the far angle approaches $\pi/2$ after taking $d\to\infty$. With $\theta_1=2\arcsin(\sin(\theta_2/2)/c)$, the corrected hypercube ratio tends to zero, while the cross-polytope ratio tends to $1/(2c^2-1)$~\cite{andoni15cp}. Thus cross-polytope LSH does not dominate the full-dimensional hypercube in all fixed-angle ratios. For example, at $c=2$ and $\theta_2=89$ degrees the limiting ratios are about $0.0994$ and $0.1450$, respectively. Turning this ratio comparison into a running-time bound still requires the table and candidate costs in Lemma~\ref{lem:lsh}.

\subsection{Convergence to the limit}

To get an idea how hypercube LSH compares to other methods when $d$ is not too large, we start by giving explicit collision probabilities for the first non-trivial case, namely $d = 2$.

\begin{proposition}[Square LSH] \label{prop:square}
For $d = 2$, $p(\theta) = 1 - \frac{2 \theta}{\pi}$ for $\theta \leq \frac{\pi}{2}$ and $p(\theta) = 0$ otherwise.
\end{proposition}

\begin{proof}
In two dimensions, two randomly rotated vectors $\vc{X}, \vc{Y}$ at angle $\theta$ can be modeled as $\vc{X} = (\cos \psi, \sin \psi)$ and $\vc{Y} = (\cos (\psi + \theta), \sin(\psi + \theta))$ for $\psi \sim [0, 2\pi)$. The conditions $\vc{X}, \vc{Y} > 0$ are then equivalent to $\psi \in (0, \frac{\pi}{2}) \cap (-\theta, \frac{\pi}{2} - \theta)$, which for $\theta < \frac{\pi}{2}$ occurs with probability $\frac{\pi/2 - \theta}{2\pi}$ over the randomness of $\psi$. As a collision can occur in any of the four quadrants, we finally multiply this probability by $4$ to obtain the stated result.
\end{proof}

Figure~\ref{fig:exp-full} depicts $p(\theta)^{1/2}$ in green, along with hyperplane LSH (blue) and the asymptotics for hypercube LSH (red). For larger $d$, computing $p(\theta)$ exactly becomes more complicated, and so instead we performed experiments to empirically obtain estimates for $p(\theta)$ as $d$ increases. These estimates are also shown in Figure~\ref{fig:exp-full}, and are based on $10^5$ trials for each $\theta$ and $d$. Observe that as $\theta \to \frac{\pi}{2}$ and/or $d$ grows larger, $p(\theta)$ decreases and the empirical estimates become less reliable. Points are omitted for cases where no successes occurred. 

\begin{figure}[!t]
\begin{center}
\includegraphics[width=\textwidth]{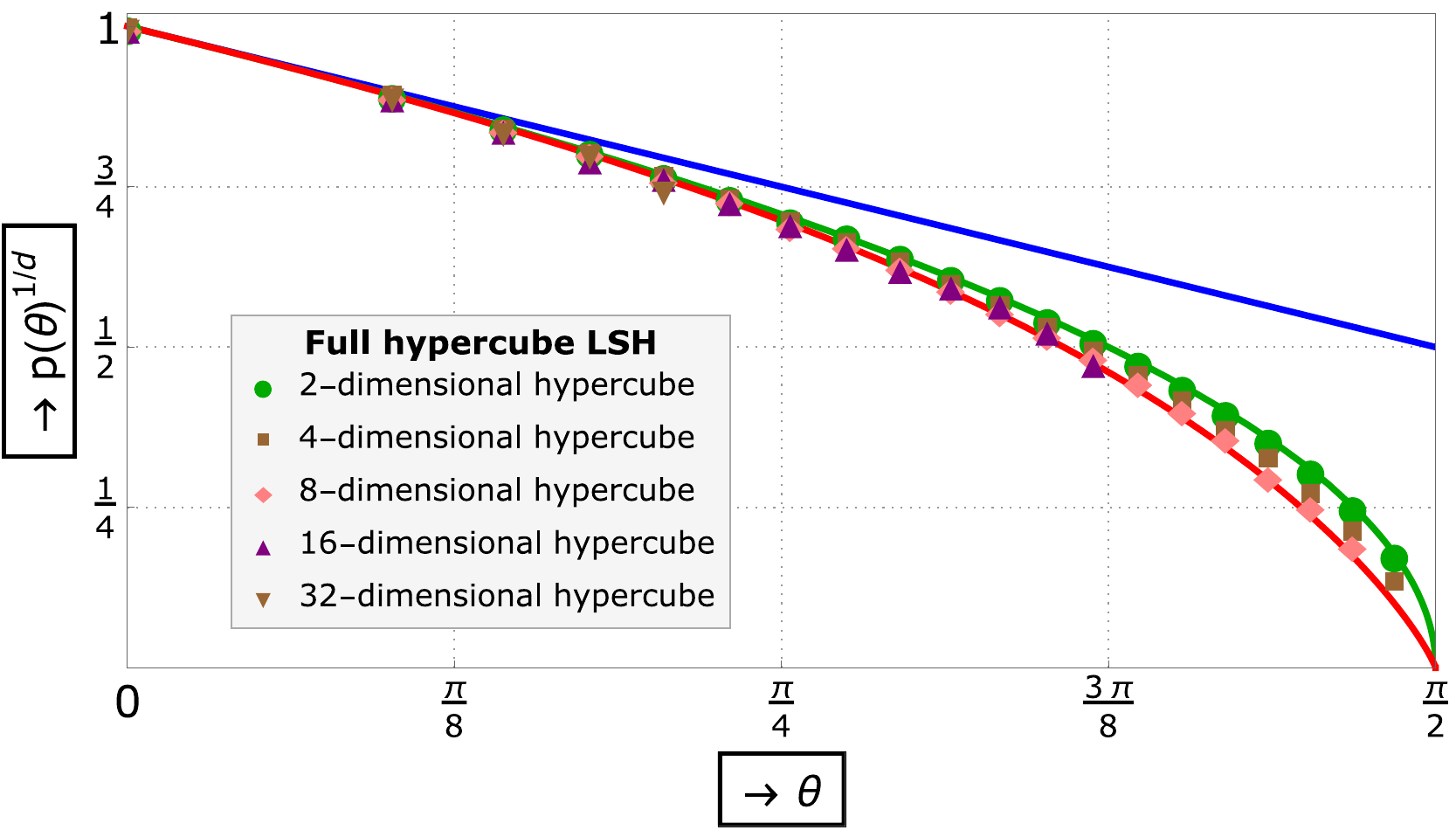}
\end{center}
\caption{Empirical collision probabilities for hypercube LSH for small $d$. The red curve denotes the asymptotic limit for hypercube LSH, the blue curve denotes hyperplane LSH, and the green curve denotes the exact case $d = 2$ from Proposition~\ref{prop:square}. \label{fig:exp-full}}
\end{figure}

The empirical estimates can be compared with the corrected asymptotic curve. In particular, the positive limiting value previously suggested as $\theta\to\pi/2$ was an artifact of the incorrect third branch. Both the corrected limiting base and the finite-dimensional collision probabilities tend to zero at this boundary. The data do not establish a general direction of convergence of $\rho$ as $d$ increases.

\paragraph*{Comparison with hyperplane/cross-polytope LSH.} Finally, \cite[Figures 1 and 2]{terasawa07} previously illustrated that among several LSH methods, the smallest values $\rho$ (for their parameter sets) are obtained with hypercube LSH with $d = 16$, achieving smaller values $\rho$ than e.g.\ cross-polytope LSH with $d = 256$. These finite-dimensional observations remain relevant, but the earlier explanation based on a positive near-orthogonal limiting base no longer applies. Cross-polytope LSH also converges slowly to its asymptotic ratio~\cite{andoni15cp}. Comparisons in moderate dimensions therefore require the finite-dimensional collision probabilities.

\subsection{Fast hashing in practice}

To further assess the practicality of hypercube LSH, recall that hashing is done as follows:
\begin{itemize}
\item Apply a uniformly random rotation $A$ to $\vc{x}$;
\item Look at the signs of $(A \vc{x})_i$. 
\end{itemize}
Theoretically, a uniformly random rotation will be rather expensive to compute, with $A$ being a real, dense matrix. As previously discussed in e.g.~\cite{achlioptas01}, it may suffice to only consider a sparse subset of all rotation matrices with a large enough amount of randomness, and as described in~\cite{andoni15cp, kennedy17} pseudo-random rotations may also be help speed up the computations in practice. Kennedy and Ward~\cite{kennedy17} prove such a fast construction for cross-polytope LSH. That result does not by itself establish the full-dimensional hypercube collision asymptotics for a fast pseudo-random transform; Theorem~\ref{thm:main} uses a uniformly random rotation.

Finally, to compare this with cross-polytope LSH, note that cross-polytope LSH in dimension $d$ partitions the space in $2d$ regions, as opposed to $2^d$ for hypercube hashing. To obtain a similar fine-grained partition of the space with cross-polytopes, one would have to concatenate $\Theta(d / \log d)$ random cross-polytope hashes, which corresponds to computing $\Theta(d / \log d)$ (pseudo-)random rotations, compared to only one rotation for hypercube LSH. We therefore expect hashing to be up to a factor $\Theta(d / \log d)$ less costly.



\section{Partial hypercube LSH}
\label{sec:partial}

Since a high-dimensional hypercube partitions the space in a large number of regions, for various applications one may only want to use hypercubes in a lower dimension $d' < d$. In those cases, one would first apply a random rotation to the data set, and then compute the hash based on the signs of the first $d'$ coordinates of the rotated data set. This corresponds to the hash family $\mathcal{H} = \{h_{A,d'}: A \in SO(d)\}$, with $h_{A,d'}$ satisfying:
\begin{align}
h_{A,d'}(\vc{x}) = (h_1(A \vc{x}), \dots, h_{d'}(A \vc{x})), \qquad \quad h_i(\vc{x}) = \begin{cases} +1, & \text{if } x_i \geq 0; \\ -1, & \text{if } x_i < 0. \end{cases} 
\end{align}
When ``projecting'' down onto the first $d'$ coordinates, observe that distances and angles are distorted: the angle between the vectors formed by the first $d'$ coordinates of $\vc{x}$ and $\vc{y}$ may not be the same as $\phi(\vc{x},\vc{y})$. The amount of distortion depends on the relation between $d'$ and $d$. Below, we will investigate how the collision probabilities $p_{d',d}(\theta)$ for partial hypercube LSH scale with $d'$ and $d$, where $p_{d',d}(\theta) = \Pr(h(\vc{x}) = h(\vc{y}) \ | \ \phi(\vc{x}, \vc{y}) = \theta)$.

\subsection{Convergence to hyperplane LSH}

First, observe that for $d' = 1$, partial hypercube LSH is \textit{equal} to hyperplane LSH, i.e.\ $p_{1,d}(\theta) = 1 - \frac{\theta}{\pi}$. For $1 < d' \ll d$, we first observe that both (partial) hypercube LSH and hyperplane LSH can be modeled by a projection onto $d'$ dimensions:
\begin{itemize}
\item Hyperplane LSH: $\vc{x} \mapsto A \vc{x}$ with $A \sim \No(0,1)^{d' \times d}$;
\item Hypercube LSH: $\vc{x} \mapsto (A^*) \vc{x}$ with $A \sim \No(0,1)^{d' \times d}$.
\end{itemize}
Here $A^*$ denotes the matrix obtained from $A$ after applying Gram--Schmidt orthogonalization and normalizing its rows to length $\sqrt d$; positive row scalings do not affect the hash. In both cases, hashing is done after the projection by looking at the signs of the projected vector. Therefore, the only difference lies in the projection, and one could ask: for which $d'$, as a function of $d$, are these projections equivalent? When is a set of random hyperplanes already (almost) orthogonal?

Jiang~\cite{jiang06} constructs such a coupling for which $\max_{i,j}|A_{i,j}-A^*_{i,j}|\to0$ in probability when $d'=o(d/\log d)$. Absolute entrywise approximation alone does not control an exponentially small collision probability. The following rate statement follows instead from the density comparison proved in Appendix~\ref{app:partial-proof}.

\begin{proposition}[Convergence to hyperplane LSH]\label{prop:partial-small}
Let $p_{d',d}(\theta)$ denote the collision probabilities for partial hypercube LSH, and let $d' = o(d / \log d)$. For each fixed $\theta\in(0,\pi)$, $p_{d',d}(\theta)^{1/d'} \to 1 - \frac{\theta}{\pi}$.
\end{proposition}

\subsection{Collision rates for a fixed fraction of the coordinates}

The earlier condition $d'=\omega(\ln n)$ does not imply equality with the full-dimensional collision rates: angle preservation with high probability does not control exponentially small collision probabilities. For example, $n=d$ and $d'=\lfloor\sqrt d\rfloor$ meet that condition but fall in the hyperplane regime of Proposition~\ref{prop:partial-small}. The following result supplies the partial rates directly.

\begin{proposition}[Collision probabilities for partial hypercube LSH]
\label{prop:partial-rate}
Let $d'/d\to\alpha\in(0,1)$, and fix $\theta\in(0,\pi)$. Put $c=\cos\theta$.
There is a unique $r\in(-1,c)$ satisfying
\begin{align}
\alpha=\frac{(c-r)\arccos(-r)}{(1-cr)\sqrt{1-r^2}}.
\label{eq:partial-stationary}
\end{align}
The collision probability satisfies
\begin{align}
\frac1d\ln p_{d',d}(\theta)
\longrightarrow
\alpha\ln\frac{\arccos(-r)}\pi
+\ln\frac{1-cr}{\sin\theta\sqrt{1-r^2}}.
\label{eq:partial-logprob}
\end{align}
\end{proposition}

The proof is given in Appendix~\ref{app:partial-proof}. This is a fixed-$\alpha$
result. In particular, one cannot substitute the full-dimensional collision
curve into a partial hash, or interchange the limits $d\to\infty$ and
$\alpha\to1$. For $\pi/3<\theta<\pi/2$, taking $\alpha\uparrow1$ after $d\to\infty$
in~\eqref{eq:partial-logprob} gives the earlier expression
$(1+\cos\theta)/(\pi\sin\theta)$ as the limiting exponential base,
whereas hashing all $d$ coordinates gives the smaller corrected base of
Theorem~\ref{thm:main}.

\subsection{Empirical collision probabilities}

To get an idea of how $p_{d',d}(\theta)$ scales with $d'$ in practice, we empirically computed several values for fixed $d = 50$. The estimates of $p_{d',50}(\theta)^{1/d'}$ are plotted in Figure~\ref{fig:exp-partial}. These data points are again based on at least $10^5$ experiments for each $d'$ and $\theta$. We expect that as $d'$ increases, the collision probabilities slowly move from hyperplane hashing towards hypercube hashing, this can also be seen in the graph -- for $d' = 2$, the estimates are almost equal to the curve for hyperplane LSH, while as $d'$ increases the curve slowly moves down towards the asymptotics for full hypercube LSH. Again, we stress that as $d'$ becomes larger, the empirical estimates become less reliable, and so we did not consider even larger values for $d'$.

\begin{figure}[!t]
\begin{center}
\includegraphics[width=\textwidth]{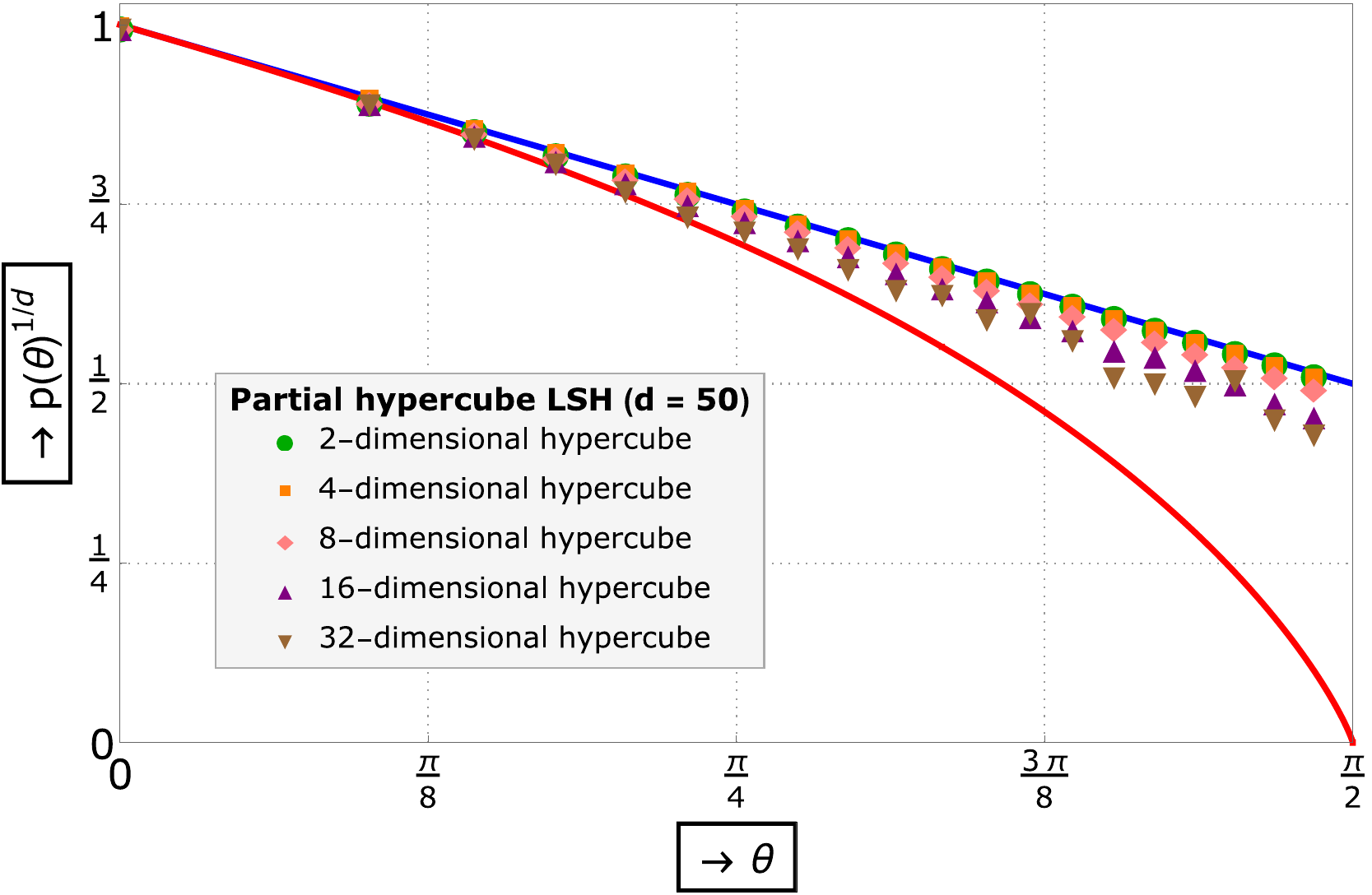}
\end{center}
\caption{Experimental values of $p_{d',50}(\theta)^{1/d'}$. The red curve denotes the full-dimensional asymptotic limit, and the blue curve denotes hyperplane LSH.\label{fig:exp-partial}}
\end{figure}

The estimates for partial hypercubes lie between the hyperplane curve and the corrected full-dimensional curve in this experiment. This observation does not establish equality of the exponential rates, nor an ordering of the corresponding ratios $\rho$.


\section{Application: Lattice sieving for the shortest vector problem} 
\label{sec:lattices}

We finally consider an explicit application for hypercube LSH, namely lattice sieving algorithms for the shortest vector problem. Given a basis $\vc{B} = \{\vc{b}_1, \dots, \vc{b}_d\} \subset \mathbb{R}^d$ of a lattice $\mathcal{L}(\vc{B}) = \{\sum_i \lambda_i \vc{b}_i: \lambda_i \in \mathbb{Z}\}$, the shortest vector problem (SVP) asks to find a shortest non-zero vector in this lattice. Various different methods for solving SVP in high dimensions are known, and currently the algorithm with the best heuristic time complexity in high dimensions is based on lattice sieving, combined with nearest neighbor searching~\cite{becker16lsf}. 

In short, lattice sieving works by generating a long list $L$ of pairwise reduced lattice vectors. For equal-length vectors $\vc x,\vc y$ of norm $R$, a pair is reduced iff $\|\vc x-\vc y\|\geq R$, equivalently $\phi(\vc x,\vc y)\geq\pi/3$. Under the usual random-direction heuristic, the list size is taken to be $n=(4/3)^{d/2+o(d)}$, the inverse exponential rate of a spherical cap of angle $\pi/3$. Therefore, if we have a list of size $n = (4/3)^{d/2 + o(d)}$, any newly sampled lattice vector can be reduced against the list many times to obtain a very short lattice vector. The time complexity of this method is dominated by doing $\mathrm{poly}(d) \cdot n$ reductions (searches for nearby vectors) with a list of size $n$. A linear search trivially leads to a heuristic complexity of $n^{2 + o(1)} = (4/3)^{d + o(d)}$ (with space $n^{1 + o(1)}$), while nearest neighbor techniques can reduce the time complexity to $n^{1 + \rho + o(1)}$ for $\rho < 1$ (increasing the space to $n^{1 + \rho + o(1)}$). For more details, see e.g.~\cite{nguyen08, laarhoven15crypto, becker16lsf}.

The full-dimensional collision rates cannot simply be substituted for the partial hash used in a sieve. Using Proposition~\ref{prop:partial-rate} instead gives the following corrected result. An outline of the sieve calculation is given in Appendix~\ref{app:lattice}.

\begin{proposition}[Complexity of lattice sieving with hypercube LSH]
\label{prop:lattice}
Under the usual random-direction heuristics for pairwise lattice sieving,
partial hypercube LSH with one random rotation per table and
$d'\approx0.20752d$ solves SVP in expected time and space at most
$2^{0.33469d+o(d)}$. The corresponding query exponent is
$\rho\approx0.61281$ for a list of size $(4/3)^{d/2+o(d)}$.
\end{proposition}

The asymptotic performance of this hypercube construction lies between those of hyperplane LSH and cross-polytope LSH.
\begin{itemize}
\item Linear search~\cite{nguyen08}: \hspace{9mm} $2^{0.4150 d + o(d)}$.
\item Hyperplane LSH~\cite{laarhoven15crypto}: \hspace{3.8mm} $2^{0.3366 d + o(d)}$.
\item \textbf{Hypercube LSH}: \hspace{8.3mm} $2^{0.3347 d + o(d)}$.
\item Spherical cap LSH~\cite{laarhoven15latincrypt}: \hspace{1.0mm} $2^{0.2972 d + o(d)}$.
\item Cross-polytope LSH~\cite{becker16cp}: \hspace{0.4mm} $2^{0.2972 d + o(d)}$.
\item Spherical LSF~\cite{becker16lsf}: \hspace{9.7mm} $2^{0.2925 d + o(d)}$.
\end{itemize}
In the experiments reported at the time of the original version (2017), the picture was almost entirely reversed~\cite{svp}. The lattice sieving method then used to solve SVP in the highest dimension ($d = 116$) used a very optimized linear search~\cite{kleinjung14}. The furthest that any nearest neighbor-based sieve had then been able to go is $d = 107$, using hypercube LSH~\cite{mariano15, mariano16pdp}\footnote{Although phrased as hyperplane LSH, the implementations from~\cite{laarhoven15crypto, mariano15, mariano16pdp} are using hypercube LSH.}. Experiments further indicated that spherical LSF only becomes competitive with hypercube LSH as $d \gtrsim 80$~\cite{becker16lsf, mariano17}, while sieving with cross-polytope LSH turned out to be rather slow compared to other methods~\cite{becker16cp, mariano16private}. Although it remains unclear which nearest neighbor method is the ``most practical'' in the application of lattice sieving, hypercube LSH is one of the main contenders.


\paragraph*{Acknowledgments.}

The author is indebted to Ofer Zeitouni for his suggestion to use results from large deviations theory, and for his many helpful comments regarding this application. The author further thanks Brendan McKay and Carlo Beenakker for their comments. The author is supported by the SNSF ERC Transfer Grant CRETP2-166734 FELICITY.


\newcommand{\etalchar}[1]{$^{#1}$}


\appendix
\renewcommand*{\theHsection}{appendix.\Alph{section}}


\section{Proof of Theorem~\ref{thm:main}}

Theorem~\ref{thm:main} will be proved through a series of lemmas, each making partial progress towards a final solution. Reading only the claims made in the lemmas may give the reader an idea how the proof is built up. Before starting the proof, we begin with a useful lemma regarding integrals of (exponentials of) quadratic forms.

\begin{lemma}[Integrating an exponential of a quadratic form in the positive quadrant] \label{lem:exp}
Let $a,b,c\in\mathbb{R}$ with $a,c<0$, $b<2\sqrt{ac}$, and $D=b^2-4ac$. Then
\begin{align}
\int_0^\infty\!\int_0^\infty e^{ax^2+bxy+cy^2}\,dx\,dy
=\begin{cases}
\dfrac{\pi+2\arctan(b/\sqrt{-D})}{2\sqrt{-D}},&D<0,\\[2ex]
-\dfrac1b,&D=0,\\[2ex]
\dfrac{\operatorname{arcosh}(-b/(2\sqrt{ac}))}{\sqrt D},&D>0.
\end{cases}
\end{align}
In the last two cases, the assumptions imply $b<0$.
\end{lemma}

\begin{proof}
Substituting $y=xs$ and then integrating over $x$ gives
\begin{align}
I=-\frac12\int_0^\infty\frac{ds}{a+bs+cs^2}.
\end{align}
The assumptions ensure $a+bs+cs^2<0$ for every $s\geq0$. For $D<0$, integration gives the stated arctangent formula, as in~\cite[Equation (3.3.16)]{abramowitz72}. For $D=0$ the denominator is a negative perfect square and the integral equals $-1/b$. For $D>0$, integrating by partial fractions gives
\begin{align}
I=\frac{1}{2\sqrt D}\ln\left(\frac{-b+\sqrt D}{-b-\sqrt D}\right)
=\frac{\operatorname{arcosh}(-b/(2\sqrt{ac}))}{\sqrt D}.
\end{align}
Conversely, the integral diverges unless $a,c<0$ and $b<2\sqrt{ac}$. A positive quadratic coefficient gives divergence near the corresponding axis; a zero coefficient gives at least a logarithmic divergence. For $a,c<0$ and $b\geq2\sqrt{ac}$, the quadratic form is nonnegative along an interior ray, giving divergence.
\end{proof}

Next, we begin by restating the collision probability between two vectors in terms of half-normal vectors.

\begin{lemma}[Towards three-dimensional large deviations] \label{lem:step1}
Let $p(\theta)$ denote the collision probability for full-dimensional hypercube hashing in dimension $d$. Let $\hat{\vc{X}},\hat{\vc{Y}}\sim\HNo(0,1)^d$ be independent, and define
\begin{align}
\vc{Z}_d=\frac1d\left(\sum_{i=1}^d\hat X_i\hat Y_i,
\sum_{i=1}^d\hat X_i^2,\sum_{i=1}^d\hat Y_i^2\right).
\end{align}
Write $Q_d(J)=\Pr(\phi(\hat{\vc{X}},\hat{\vc{Y}})\in J)$, and let $g_d$ denote the angle density for two independent standard normal vectors. Then, for every angular interval $J$,
\begin{align}
Q_d(J)=2^d\int_J p(t)g_d(t)\,dt. \label{eq:1}
\end{align}
\end{lemma}

\begin{proof}
For independent standard normal vectors, the event that both lie in the positive orthant has probability $2^{-2d}$. Conditional on this event, they are independent half-normal vectors. Each of the $2^d$ orthants contributes equally to a collision. Applying this symmetry and conditioning on the angle gives~\eqref{eq:1}.
\end{proof}

\begin{lemma} [Applying Cram\'er's theorem to $\vc{Z}_d$] \label{lem:ge2}
Let $\vc{Z}_d$ be as in Lemma~\ref{lem:step1}, and let $\Lambda$ and $\Lambda^*$ be as in Section~\ref{sec:pre}. Then $\vc{0}$ lies in the interior of $\mathcal D_\Lambda$, so Lemma~\ref{lem:ge} applies. In particular, at every point $\vc z=(xy\cos\theta,x^2,y^2)$ with $x,y>0$ and $0<\theta<\pi/2$,
\begin{align}
\lim_{\delta\downarrow0}\liminf_{d\to\infty}\frac1d\ln\Pr(\|\vc Z_d-\vc z\|<\delta)
&=-\Lambda^*(\vc z),\\
\lim_{\delta\downarrow0}\limsup_{d\to\infty}\frac1d\ln\Pr(\|\vc Z_d-\vc z\|<\delta)
&=-\Lambda^*(\vc z).
\end{align}
\end{lemma}

\begin{proof}
The moment generating function is finite in a neighborhood of $\vc0$, since a sufficiently small quadratic perturbation of the half-normal density remains integrable. The points in the statement lie in the interior of the convex support $\{\vc z:z_2,z_3\geq0,\ 0\leq z_1\leq\sqrt{z_2z_3}\}$. The rate function is finite and continuous there; applying the open- and closed-ball bounds of Lemma~\ref{lem:ge} and letting $\delta\downarrow0$ proves the claim.
\end{proof}

Essentially, all that remains now is computing $\Lambda^*$ at the appropriate point $\vc{z}$. To continue, we first compute the logarithmic moment generating function $\Lambda = \Lambda_d$ of $\vc{Z}_d$:

\begin{lemma} [Computing $\Lambda$]
Let $\vc Z_d$ be as before, $t_2=1-2\lambda_2$, $t_3=1-2\lambda_3$, and $D=\lambda_1^2-t_2t_3$. The full domain of the logarithmic moment generating function is
\begin{align}
\mathcal D_\Lambda=\{\vc\lambda\in\mathbb R^3:t_2,t_3>0,\ \lambda_1<\sqrt{t_2t_3}\}.
\end{align}
On this domain,
\begin{align}
e^{\Lambda(\vc\lambda)}=\begin{cases}
\dfrac{\pi+2\arctan(\lambda_1/\sqrt{-D})}{\pi\sqrt{-D}},&D<0,\\[2ex]
-\dfrac2{\pi\lambda_1},&D=0,\\[2ex]
\dfrac{2\operatorname{arcosh}(-\lambda_1/\sqrt{t_2t_3})}{\pi\sqrt D},&D>0.
\end{cases}
\end{align}
\end{lemma}

\begin{proof}
By the definition of the LMGF, we have:
\begin{align}
\Lambda(\vc{\lambda}) &= \ln \expn_{\hat{X}_1, \hat{Y}_1 \sim \HNo(0,1)} \left[\exp \left(\lambda_1 \hat{X}_1 \hat{Y}_1 + \lambda_2 \hat{X}_1^2 + \lambda_3 \hat{Y}_1^2 \right) \right].
\end{align}
We next compute the inner expectation over the random variables $\hat{X}_1, \hat{Y}_1$, by writing out the double integral over the product of the argument with the densities of $\hat{X}_1$ and $\hat{Y}_1$.
\begin{align}
& \expn_{X_1, Y_1} \left[\exp \left(\lambda_1 X_1 Y_1 + \lambda_2 X_1^2 + \lambda_3 Y_1^2 \right) \right] \\
 &= \int_0^{\infty} \sqrt{\frac{2}{\pi}} \, \exp\left(-\frac{x^2}{2}\right) d x \int_0^{\infty} \sqrt{\frac{2}{\pi}} \, \exp\left(-\frac{y^2}{2}\right) dy \, \exp \left( \lambda_1 x y + \lambda_2 x^2 + \lambda_3 y^2 \right) \\
 &= \frac{2}{\pi} \int_0^{\infty} \int_0^{\infty} \exp \left( \lambda_1 x y + \left(\lambda_2 - \tfrac{1}{2}\right) x^2 + \left(\lambda_3 - \tfrac{1}{2}\right) y^2 \right) \, dx \, dy \, .
\end{align}
Applying Lemma~\ref{lem:exp} with $(a, b, c) = (\lambda_2  - \tfrac{1}{2}, \lambda_1, \lambda_3 - \tfrac{1}{2})$ yields the claimed expression for $\Lambda$, as well as the bounds stated in $\mathcal{D}_{\Lambda}$ which are necessary for the expectation to be finite.
\end{proof}

We now continue with computing the Fenchel-Legendre transform of $\Lambda$, which involves a rather complicated maximization (supremum) over $\vc{\lambda} \in \mathbb{R}^3$. The following lemma makes a first step towards computing this supremum.

For $\lambda_1\ne0$, put $\beta=\sqrt{t_2t_3}/|\lambda_1|$. Its range is $\beta>1$ when $\lambda_1>0$, and $\beta>0$ when $\lambda_1<0$. Define
\begin{align}
K(\lambda_1,\beta)=\begin{cases}
\dfrac{\pi+2\arctan(\lambda_1/(|\lambda_1|\sqrt{\beta^2-1}))}{\sqrt{\beta^2-1}},&\beta>1,\\[2ex]
2,&\beta=1,\ \lambda_1<0,\\[1ex]
\dfrac{2\operatorname{arcosh}(1/\beta)}{\sqrt{1-\beta^2}},&0<\beta<1,\ \lambda_1<0.
\end{cases}
\end{align}
Thus $e^{\Lambda(\vc\lambda)}=K(\lambda_1,\beta)/(\pi|\lambda_1|)$.

\begin{lemma}[Computing $\Lambda^*(\vc{z})$ -- General form]
Let $z_2,z_3>0$. Then
\begin{align}
\Lambda^*(\vc z)=\ln\pi+\frac{z_2+z_3}{2}
+\sup_{\lambda_1,\beta}\big\{\lambda_1z_1-|\lambda_1|\beta\sqrt{z_2z_3}
+\ln|\lambda_1|-\ln K(\lambda_1,\beta)\big\},
\end{align}
where the supremum ranges over the two domains specified above.
\end{lemma}

\begin{proof}
For fixed $\lambda_1$ and $u=t_2t_3>0$, the dependence of $\langle\vc\lambda,\vc z\rangle-\Lambda(\vc\lambda)$ on $t_2,t_3$ is only through $-(t_2z_2+t_3z_3)/2$. The maximum therefore occurs at $t_2=\sqrt{uz_3/z_2}$ and $t_3=\sqrt{uz_2/z_3}$, with value $-\sqrt{uz_2z_3}$. Substituting $u=\beta^2\lambda_1^2$ and the preceding formula for $e^\Lambda$ gives the result. The case $\lambda_1=0$ is obtained by continuity and does not change the supremum.
\end{proof}

For $\vc z=(xy\cos\theta,x^2,y^2)$ with $x,y>0$ and $0<\theta<\pi/2$, this becomes
\begin{align}
\Lambda^*(xy\cos\theta,x^2,y^2)=\ln\pi+\frac{x^2+y^2}{2}
+\sup_{\lambda_1,\beta}\big\{(\lambda_1\cos\theta-|\lambda_1|\beta)xy
+\ln|\lambda_1|-\ln K(\lambda_1,\beta)\big\}. \label{eq:2}
\end{align}
We treat the two signs of $\lambda_1$ separately, writing
\begin{align}
\Lambda^*(\vc z)=\max\{\Lambda^*_+(\vc z),\Lambda^*_-(\vc z)\}, \label{eq:split}
\end{align}
where the subscripts denote restriction to $\lambda_1>0$ and $\lambda_1<0$, respectively.

\begin{lemma} [Computing $\Lambda^*(\vc{z})$ for positive $\lambda_1$]
Let $\vc{z} = (x y \cos \theta, x^2, y^2)$ with $x, y > 0$ and $\theta \in (0, \frac{\pi}{2})$. For $\theta \in (0, \arccos \frac{2}{\pi})$, let $\beta_0 = \beta_0(\theta) \in (1, \infty)$ be the unique solution to~\eqref{eq:beta}.
Then the Fenchel-Legendre transform $\Lambda^*$ at $\vc{z}$, restricted to $\lambda_1 > 0$, satisfies 
\begin{align}
\Lambda^*_{+}(\vc{z}) &= \frac{x^2}{2} + \frac{y^2}{2} - 1 - \ln(xy) + \begin{cases} 
\ln \left(\dfrac{\pi \beta_0 (\beta_0 \cos \theta - 1)}{2 (\beta_0 - \cos \theta)^2 }\right), & \text{if } \theta \in (0, \arccos \frac{2}{\pi}); \\[3ex]
0, & \text{if } \theta \in [\arccos \frac{2}{\pi}, \frac{\pi}{2}).
\end{cases}
\end{align}
\end{lemma}

\begin{proof}
Substituting $\lambda_1 > 0$ into~\eqref{eq:2}, we obtain:
\begin{align}
& \Lambda^*_{+}(x y \cos \theta, x^2, y^2) = \ln \pi + \frac{x^2}{2} + \frac{y^2}{2} + \sup_{\substack{\lambda_1 > 0 \\ \beta > 1}} \Big\{g_+(\lambda_1, \beta)\Big\}, \\
& g_+(\lambda_1, \beta) = (\cos \theta - \beta) \lambda_1 x y + \frac{\ln (\beta^2 - 1)}{2} + \ln \lambda_1 - \ln \Big(\pi + 2 \arctan\Big(\frac{1}{\sqrt{\beta^2 - 1}}\Big)\Big).
\end{align}
Differentiating w.r.t. $\lambda_1$ gives $(\cos \theta - \beta) x y + \frac{1}{\lambda_1}$. Recall that $\beta > 1 > \cos \theta$. For $\lambda_1 \to 0^+$ the derivative is therefore positive, for $\lambda_1 \to \infty$ it is negative, and there is a global maximum at the only root $\lambda_1 = 1 / ((\beta - \cos \theta) x y)$. In that case, the expression further simplifies 
and we can pull out more terms that do not depend on $\beta$, to obtain:
\begin{align}
& \Lambda^*_{+}(x y \cos \theta, x^2, y^2) = \ln \pi + \frac{x^2}{2} + \frac{y^2}{2} - 1 - \ln(xy) + \sup_{\beta > 1} \Big\{g_+(\beta)\Big\}, \\
& g_+(\beta) = \ln \left( \frac{\sqrt{\beta^2 - 1}}{(\beta - \cos \theta) \left(\pi + 2 \arcsin \frac{1}{\beta}\right)}\right) = \ln h_+(\beta).
\end{align}
Here we used the identity $\arctan(1/\sqrt{\beta^2-1}) = \arcsin(1/\beta)$. 
Now, for $\beta \to 1^+$ we have $h_+(\beta) \to 0^+$, while for $\beta \to \infty$, we have 
\begin{align}
h_+(\beta) = \frac{1}{\pi} + \frac{1}{\pi \beta} \left(\cos \theta - \frac{2}{\pi}\right) + O\left(\frac{1}{\beta^2}\right).
\end{align}
In other words, if $\cos \theta \leq \frac{2}{\pi}$ or $\theta \geq \arccos \frac{2}{\pi}$, we have $h_+(\beta) \to (\tfrac{1}{\pi})^-$ (the second order term is negative for $\cos \theta = \frac{2}{\pi}$), while for $\theta < \arccos \frac{2}{\pi}$ we approach the same limit from above as $h_+(\beta) \to (\tfrac{1}{\pi})^+$. For $\theta < \arccos \frac{2}{\pi}$ there is a non-trivial maximum at some value $\beta = \beta_0 \in (1, \infty)$, while for $\theta \geq \arccos \frac{2}{\pi}$, we can see from the derivative $h_+'(\beta)$ that $h_+(\beta)$ is strictly increasing on $(1, \infty)$, and the supremum is attained at $\beta \to \infty$. We therefore obtain two different results, depending on whether $\theta < \arccos \frac{2}{\pi}$ or $\theta \geq \arccos \frac{2}{\pi}$.

\textbf{Case 1}: $\arccos \tfrac{2}{\pi} \leq \theta < \tfrac{\pi}{2}$. The supremum is attained in the limit of $\beta \to \infty$, which leads to $h_+(\beta) \to \frac{1}{\pi}$ and the stated expression for $\Lambda^*_{+}(x y \cos \theta, x^2, y^2)$.

\textbf{Case 2}: $0 < \theta < \arccos \tfrac{2}{\pi}$. In this case there is a non-trivial maximum at some value $\beta = \beta_0$, namely there where the derivative $h_+'(\beta_0) = 0$. After computing the derivative, eliminating the (positive) denominator and rewriting, this condition is equivalent to~\eqref{eq:beta}.
This allows us to rewrite $g$ and $\Lambda^*$ in terms of $\beta_0$, by substituting the given expression for $\arcsin\left(\frac{1}{\beta_0}\right)$, which ultimately leads to the stated formula for $\Lambda^*_+$.
\end{proof}

\begin{lemma} [Computing $\Lambda^*(\vc{z})$ for negative $\lambda_1$]
Let $\vc z=(xy\cos\theta,x^2,y^2)$, with $x,y>0$ and $0<\theta<\pi/2$. Let $\beta_1$ and $\beta_2$ be as in Theorem~\ref{thm:main}. Then
\begin{align}
\Lambda^*_-(\vc z)=\frac{x^2+y^2}{2}-1-\ln(xy)
+\begin{cases}
0,&0<\theta\leq\arccos(2/\pi),\\[2ex]
\ln\!\left(\dfrac{\pi\beta_1(1+\beta_1\cos\theta)}{2(\beta_1+\cos\theta)^2}\right),
&\arccos(2/\pi)<\theta<\pi/3,\\[2ex]
\ln\!\left(\dfrac{\pi\beta_2(1+\beta_2\cos\theta)}{2(\beta_2+\cos\theta)^2}\right),
&\pi/3\leq\theta<\pi/2.
\end{cases}
\end{align}
At $\theta=\pi/3$, the last expression is interpreted with $\beta_2=1$.
\end{lemma}

\begin{proof}
For $\lambda_1<0$, write $K(\lambda_1,\beta)=2A(\beta)$, where
\begin{align}
A(\beta)=\begin{cases}
\dfrac{\arccos(1/\beta)}{\sqrt{\beta^2-1}},&\beta>1,\\[2ex]
1,&\beta=1,\\[1ex]
\dfrac{\operatorname{arcosh}(1/\beta)}{\sqrt{1-\beta^2}},&0<\beta<1.
\end{cases}
\end{align}
Substituting into~\eqref{eq:2} and maximizing over $\lambda_1<0$ gives the unique maximum at $\lambda_1=-1/((\cos\theta+\beta)xy)$, and hence
\begin{align}
\Lambda^*_-(\vc z)=\frac{x^2+y^2}{2}-1-\ln(xy)
+\ln\frac\pi2-\ln\inf_{\beta>0}\big\{(\beta+\cos\theta)A(\beta)\big\}. \label{eq:negative-inf}
\end{align}
Put $c=\cos\theta$. For $\beta>1$, substitute $\beta=1/\cos v$, $0<v<\pi/2$. The objective becomes $v(1+c\cos v)/\sin v$. Its stationary-point equation is
\begin{align}
c=\frac{\sin v-v\cos v}{v-\sin v\cos v}.
\end{align}
The ratio on the right increases strictly from $1/2$ to $2/\pi$: both numerator and denominator vanish at zero, and the ratio of their derivatives is $v/(2\sin v)$, which increases strictly. Thus a unique minimum with $\beta>1$ exists exactly when $1/2<c<2/\pi$, and its stationary equation is~\eqref{eq:beta} for $\beta_1$.

For $0<\beta<1$, substitute $\beta=1/\cosh u$, $u>0$. The objective becomes $u(1+c\cosh u)/\sinh u$, and its stationary-point equation is
\begin{align}
c=\frac{u\cosh u-\sinh u}{\sinh u\cosh u-u}. \label{eq:hyperbolic-angle}
\end{align}
This ratio decreases strictly from $1/2$ to $0$: the ratio of derivatives is $u/(2\sinh u)$, which decreases strictly. Thus a unique minimum with $0<\beta<1$ exists exactly when $0<c<1/2$, and its stationary equation is~\eqref{eq:beta2}. The derivative signs also show that at $c=1/2$ the minimum occurs at $\beta=1$, while for $c\geq2/\pi$ the infimum occurs at $\beta\to\infty$ and equals $\pi/2$.

At either nontrivial stationary point, substitution of its parameter equation yields
\begin{align}
(\beta+c)A(\beta)=\frac{(\beta+c)^2}{\beta(1+\beta c)}.
\end{align}
Substitution into~\eqref{eq:negative-inf} proves the claim.
\end{proof}

\begin{proof}[Proof of Theorem~\ref{thm:main}]
Combining the preceding two lemmas with~\eqref{eq:split} gives
\begin{align}
\Lambda^*(xy\cos\theta,x^2,y^2)
=\frac{x^2+y^2}{2}-1-\ln(xy)+I(\theta),
\end{align}
where $I$ is continuous for $0<\theta<\pi/2$. The first three terms have minimum zero, attained at $x=y=1$.

To pass from neighborhood probabilities to collision probabilities at an exact angle, let $J\subset(0,\pi/2)$ be a small interval. The norms of independent half-normal vectors have independent $\chi_d$ distributions and are independent of their directions. Intersecting the angular event with $\|\hat{\vc X}\|^2/d,\|\hat{\vc Y}\|^2/d\in[1-\eta,1+\eta]$ therefore multiplies its probability by a factor tending to one. On this restricted event, the angular map $\vc z\mapsto\arccos(z_1/\sqrt{z_2z_3})$ is continuous on a compact set. Lemma~\ref{lem:ge}, the computed rate, and then $\eta\downarrow0$ give
\begin{align}
\lim_{d\to\infty}\frac1d\ln Q_d(J)=-\inf_{t\in J}I(t),
\end{align}
with the same infimum for the interior and closure of $J$.

For every fixed rotation, a collision between $\vc x$ and $\vc x\cos t+\vc z\sin t$, with $\vc x,\vc z$ orthonormal and $0\leq t<\pi/2$, persists on an initial interval of $t$: each coordinate with opposite signs imposes one upper bound on $\tan t$. Hence $p(t)$ is nonincreasing. Applying~\eqref{eq:1} to $[\theta-\varepsilon,\theta]$ and $[\theta,\theta+\varepsilon]$, and using Lemma~\ref{lem:sphericalcap}, bounds $d^{-1}\ln p(\theta)$ between the corresponding angular rate differences. Letting $\varepsilon\downarrow0$ gives
\begin{align}
\lim_{d\to\infty}\frac1d\ln p(\theta)=-I(\theta)-\ln(2\sin\theta).
\end{align}
Substituting the values of $I$ yields the three nonzero branches. At $\theta=\pi/2$ and beyond, vectors in the same open orthant have positive inner product, so a collision has probability zero. Finally $p(0)=1$, and the interior boundary values follow by continuity and the same angular-interval argument.
For the corrected third branch,~\eqref{eq:hyperbolic-angle} gives $u\to\infty$ as $\theta\uparrow\pi/2$ and
\begin{align}
p(\theta)^{1/d}\longrightarrow
\frac{u(1+\cos\theta\cosh u)}{\pi\sin\theta\sinh u}\longrightarrow0,
\end{align}
where the first limit is $d\to\infty$ at each fixed $\theta$, and the second is $\theta\uparrow\pi/2$.

\end{proof}


\section{Proof of Proposition~\ref{prop:lattice}}
\label{app:lattice}

Write
\begin{align}
c_n=\frac12\log_2\frac43,\qquad
n=2^{c_nd+o(d)}.
\end{align}
We first analyze pair search on a list of independent uniform sphere points.
Each table uses one random rotation and the signs of its first
$m=d'$ coordinates. For any fixed rotation, all $2^m$ buckets have spherical
measure $2^{-m}$. Thus two independent uniform sphere points collide with
probability exactly $2^{-m}$, and the expected number of colliding unordered
pairs in one table is
\begin{align}
\binom n2\,2^{-m}.
\end{align}
This expectation averages over the complete angular distribution; it does
not replace that distribution by a single far angle.

Put $p_m=p_{m,d}(\pi/3)$. The partial-hypercube collision probability is
nonincreasing in the angle, as in the monotonicity argument in Appendix~A.
Consequently, every pair at angle at most $\pi/3$ collides in each independent
table with probability at least $p_m$. With
\begin{align}
L=\left\lceil\frac{C\ln n}{p_m}\right\rceil
\end{align}
tables, its probability of being missed is at most $n^{-C}$.
A union bound over all pairs shows that all pairs at angle at most $\pi/3$
are found with probability at least $1-n^{2-C}$.

Computing the rotated vectors, forming the buckets, and checking every
colliding pair takes polynomial time in $d$ per vector or pair. Sorting the
buckets introduces at most another polynomial factor, since $\ln n=O(d)$.
The expected total time is therefore
\begin{align}
\mathrm{poly}(d)\,L\bigl(n+n^2 2^{-m}\bigr).
\label{eq:partial-sieve-cost}
\end{align}
Keeping all tables uses at most $\mathrm{poly}(d)nL$ space. For pair search
on a fixed list, the tables may instead be processed successively with
$\mathrm{poly}(d)n$ working space, in addition to the output. The expected
number of output pairs is $n^{1+o(1)}$, because the spherical cap of angle
$\pi/3$ has measure $(3/4)^{d/2+o(d)}=n^{-1+o(1)}$.

We now substitute the actual partial-hypercube rates from
Proposition~\ref{prop:partial-rate}. For $m/d\to\alpha\in(0,1)$, write
\begin{align}
p_{m,d}(\pi/3)=\exp\bigl(-dJ_\alpha(\pi/3)+o(d)\bigr).
\end{align}
Then~\eqref{eq:partial-sieve-cost} has exponential rate
\begin{align}
\frac{\ln T}{d}
=c_n\ln2+J_\alpha(\pi/3)
+\max\{0,c_n\ln2-\alpha\ln2\}+o(1).
\label{eq:partial-sieve-rate}
\end{align}
For this construction the minimum occurs at $\alpha=c_n$, corresponding to
$m=\lceil\log_2n\rceil$. To see this, let $r$ be the stationary parameter
in Proposition~\ref{prop:partial-rate}, with $c=\cos(\pi/3)=1/2$.
Differentiating the optimized rate with respect to $\alpha$ gives
\begin{align}
\frac{\partial J_\alpha(\pi/3)}{\partial\alpha}
=\ln\frac{\pi}{\arccos(-r)}>0.
\end{align}
The stationary equation has $r=0$ at $\alpha=\pi/4$, and $r>0$ for
$0<\alpha\leq c_n<\pi/4$. Hence the derivative of the rate in
\eqref{eq:partial-sieve-rate} is
$\ln(\pi/\arccos(-r))-\ln2<0$ below $c_n$, and is positive above $c_n$.

At $\alpha=c_n$, the unique relevant solution of
\begin{align}
c_n=\frac{(1/2-r)\arccos(-r)}{(1-r/2)\sqrt{1-r^2}}
\end{align}
lies in $0.42659524<r<0.42659525$. Substituting it into the partial collision rate
gives
\begin{align}
\frac{J_{c_n}(\pi/3)}{\ln2}
&\approx0.12717028,\\
\rho=\frac{J_{c_n}(\pi/3)}{c_n\ln2}
&\approx0.61281344,\\
c_n+\frac{J_{c_n}(\pi/3)}{\ln2}
&\approx0.33468903.
\end{align}
This proves the stated random-list pair-search cost. Under the usual
heuristic model of sieve vectors as random directions, the polynomial
number of such pair-search stages or list queries does not change the
exponent, yielding the SVP time and space upper bounds in the proposition.

For completeness, conditioning a random pair on a collision biases its
angle. In the first $m$ coordinates the signs agree, so the products of
independent standard normal coordinates have conditional mean $2/\pi$;
in the remaining coordinates their mean is zero. The normalized norms
concentrate at one. Thus the conditional cosine concentrates at
$2\alpha/\pi$, which at $\alpha=c_n$ corresponds to an angle of about
$82.4084$ degrees. The typical candidates are therefore not orthogonal.
The previous estimate $2^{0.3222d+o(d)}$ used full-dimensional rates in
place of these partial rates and does not follow from that calculation.

\section{Proofs for partial hypercube LSH}
\label{app:partial-proof}

\subsection{A fixed fraction of the coordinates}

\begin{proof}[Proof of Proposition~\ref{prop:partial-rate}]
Write $m=d'$. Let $X_i,Y_i$ be independent half-normal variables for $i\leq m$
and independent standard normal variables for $i>m$, and put
\begin{align}
\vc Z_d=\frac1d\left(\sum_i X_iY_i,\sum_i X_i^2,\sum_i Y_i^2\right).
\end{align}
As in Lemma~\ref{lem:step1}, conditioning the first $m$ coordinates of both
vectors to be positive gives, for every angular interval $J$,
\begin{align}
Q_d(J):=\Pr(\phi(\vc X,\vc Y)\in J)
=2^m\int_Jp_{m,d}(t)g_d(t)\,dt,
\label{eq:partial-angle-identity}
\end{align}
where $g_d$ is the angle density of two independent standard normal vectors.

For $t_2=1-2\lambda_2$, $t_3=1-2\lambda_3$, and
$r=\lambda_1/\sqrt{t_2t_3}$, the limiting logarithmic moment generating
function is
\begin{align}
\Lambda_\alpha(\vc\lambda)
=-\frac12\ln(t_2t_3-\lambda_1^2)+\alpha\ln C(r),
\qquad C(r)=\frac{2\arccos(-r)}\pi.
\label{eq:partial-mgf}
\end{align}
Its domain is $t_2,t_3>0$ and $|r|<1$. Indeed, the unrestricted normal
coordinates have moment generating function $(t_2t_3-\lambda_1^2)^{-1/2}$.
The half-normal coordinates have this same function multiplied by $C(r)$,
as follows from Lemma~\ref{lem:exp}. The positive fraction $1-\alpha$ of
unrestricted coordinates requires a positive-definite quadratic form, even
when $\lambda_1<0$.

Apply Lemma~\ref{lem:ge} separately to the two independent blocks of
coordinates. Their product large-deviation bounds and contraction under the
weighted sum give rate function $\Lambda_\alpha^*$ for $\vc Z_d$.
Equivalently, the upper bound follows from the moment generating function,
and a matching local lower bound at a stationary tilt follows by exponential
tilting and the law of large numbers in each block. We now exhibit that tilt
at every point needed below.

For $\vc z=(xyc,x^2,y^2)$ with $x,y>0$, fix $r$ and $s=\sqrt{t_2t_3}$.
The maximum over $t_2,t_3$ occurs at $t_2=sy/x$ and $t_3=sx/y$.
The remaining maximum over $s>0$ occurs at $s=1/(xy(1-cr))$. Thus
\begin{align}
\Lambda_\alpha^*(xyc,x^2,y^2)
&=\frac{x^2+y^2}{2}-1-\ln(xy)+K_\alpha(c),\\
K_\alpha(c)
&=\max_{-1<r<1}\left\{
\frac12\ln(1-r^2)-\ln(1-cr)-\alpha\ln C(r)\right\}.
\label{eq:partial-rate-optimization}
\end{align}
For $-1<c<1$ and $0<\alpha<1$, the argument tends to $-\infty$ at both
endpoints $r=\pm1$, so a maximizer exists in the interior. Differentiating
gives~\eqref{eq:partial-stationary}, which also implies $r<c$.
The maximizer is unique: $\Lambda_\alpha$ is strictly convex, because its
Hessian is a positive weighted sum of covariance matrices of
$(XY,X^2,Y^2)$ under the two tilted laws. Every nonzero linear combination
of these three variables has positive variance under the unrestricted
normal law. Each stationary $r$, together with the optimizing $s,t_2,t_3$,
would give a stationary point of
$\langle\vc\lambda,\vc z\rangle-\Lambda_\alpha(\vc\lambda)$, which is strictly
concave. Hence the root in~\eqref{eq:partial-stationary} is unique.

The first three terms in~\eqref{eq:partial-rate-optimization} are minimized
at $x=y=1$, with value zero. The norms of $\vc X$ and $\vc Y$ have
independent $\chi_d$ distributions and are independent of their directions:
conditioning signs is homogeneous and does not affect the radial laws.
Consequently, the compact restriction and angular-interval argument in the
proof of Theorem~\ref{thm:main} gives
\begin{align}
\lim_{d\to\infty}\frac1d\ln Q_d(J)
=-\inf_{t\in J}K_\alpha(\cos t)
\end{align}
for sufficiently small intervals $J\subset(0,\pi)$. The rate is continuous
there, by the unique interior maximum in~\eqref{eq:partial-rate-optimization}.

For completeness, the monotonicity used to recover an exact angle also
holds for partial hashing on all of $(0,\pi)$. For a fixed orthonormal pair
$\vc x,\vc z$, a collision between $\vc x$ and
$\vc x\cos t+\vc z\sin t$ in coordinate $i$ is equivalent, outside null
boundaries, to $\cot t>-z_i/x_i$. Since $\cot t$ decreases on $(0,\pi)$,
the intersection of the first $m$ collision conditions is an initial
interval. Therefore $p_{m,d}(t)$ is nonincreasing. Using the two adjacent
intervals in~\eqref{eq:partial-angle-identity}, then letting their lengths
tend to zero, gives
\begin{align}
\lim_{d\to\infty}\frac1d\ln p_{m,d}(\theta)
=-\alpha\ln2-\ln\sin\theta-K_\alpha(c).
\end{align}
Substituting $C(r)=2\arccos(-r)/\pi$ yields~\eqref{eq:partial-logprob}.
\end{proof}

To verify the boundary behavior mentioned after the proposition, let
$c<1/2$ and take $\alpha\uparrow1$. A stationary point with
$\alpha=1$, written $r=-\cos v$ with $0<v<\pi$, would require
\begin{align}
c=\frac{\sin v-v\cos v}{v-\sin v\cos v}>\frac12.
\end{align}
The inequality follows because the ratio of the derivatives of numerator
and denominator is $v/(2\sin v)$, which increases strictly from $1/2$.
Thus no stationary point remains in the interior, and the maximum tends
to $r=-1$. The stationary equation further gives
$1-\alpha\sim(1-2c)v^2/(3(1+c))$, so $(1-\alpha)\ln v\to0$.
Since $\arccos(-r)/\sqrt{1-r^2}\to1$, substitution into
\eqref{eq:partial-logprob} gives the stated exponential base
$(1+c)/(\pi\sqrt{1-c^2})$.

\subsection{A vanishing fraction of the coordinates}

\begin{proof}[Proof of Proposition~\ref{prop:partial-small}]
The following argument proves the assertion for $m=d'=o(d)$, and hence
under its stated assumption $d'=o(d/\log d)$. Fix $\theta\in(0,\pi)$.
Let $\vc v,\vc w$ be a uniformly random orthonormal pair, and let $X$ be
$\sqrt d$ times the matrix formed by their first $m$ rows. The density of
$X$ is proportional to
\begin{align}
\det(I_2-X^TX/d)^{(d-m-3)/2}
\end{align}
on $X^TX<dI_2$; its normalizing constant is given by the orthogonal-block
density in~\cite[Lemma 2.5]{jiang06}. Relative to an $m\times2$ matrix of
independent standard normal variables, the density ratio is
\begin{align}
R(X)&=C_{d,m}\exp\!\left(\frac{\operatorname{tr}(X^TX)}2\right)
\det(I_2-X^TX/d)^{(d-m-3)/2},\\
C_{d,m}&=\left(\frac2d\right)^m
\frac{\Gamma(d/2)\Gamma((d-1)/2)}
{\Gamma((d-m)/2)\Gamma((d-m-1)/2)}.
\end{align}
Set $R=0$ outside its support. Stirling's formula gives
$\ln C_{d,m}=O((m^2+m)/d)$. If $\lambda$ is either eigenvalue of $X^TX$,
the contribution of that eigenvalue to $\ln R-\ln C_{d,m}$ is
\begin{align}
\frac\lambda2+\frac{d-m-3}{2}\ln(1-\lambda/d).
\end{align}
This function has its maximum at $\lambda=m+3$, showing that globally
$\ln R\leq O((m^2+m)/d)$. For any fixed $K>2$, expanding the logarithm on
$\operatorname{tr}(X^TX)\leq Km$ also gives
$\ln R\geq-O_K((m^2+m)/d)$ for sufficiently large $d$.

Let $E$ be the event that the two entries
$X_{i,1}$ and $X_{i,1}\cos\theta+X_{i,2}\sin\theta$ have the same sign for
every row $i$. Under the independent normal law,
$\Pr(E)=(1-\theta/\pi)^m$. Moreover, $E$ depends only on the directions
of the rows, so $S=\operatorname{tr}(X^TX)\sim\chi^2_{2m}$ is independent
of $E$. The density-ratio bounds therefore imply
\begin{align}
e^{-O_K((m^2+m)/d)}(1-\theta/\pi)^m\Pr(S\leq Km)
&\leq p_{m,d}(\theta)\\
&\leq e^{O((m^2+m)/d)}(1-\theta/\pi)^m.
\end{align}
For $m\to\infty$, $\Pr(S\leq Km)\to1$ and $(m^2+m)/d=o(m)$.
For bounded $m$, convergence of the fixed-dimensional density to the normal
density gives the same conclusion directly. Taking $m$-th roots proves
the proposition.
\end{proof}

\end{document}